\documentclass[10pt]{article}
\usepackage{times}

\usepackage{amsmath,amssymb,graphicx}
\usepackage{wrapfig,empheq}
\usepackage{hyperref}

\newtheorem{theorem}{Theorem}
\newtheorem{thm}[theorem]{Theorem}
\newtheorem{lemma}[theorem]{Lemma}

\newtheorem{prop}[theorem]{Proposition}

\def\R{\ensuremath{\Bbb R}}

\newenvironment{proof}{\begingroup\Proof}{\qed\endgroup}
\def\Proof{\noindent{\bf Proof\/:}\nobreak}
\def\qed{\unskip~{\vrule height 4pt depth 0pt width 4pt}\medbreak}

\newcommand{\etal}{\textit{et al.}}

\title{Angle Optimization of Graphs Embedded in the Plane}
\author{Sergey Bereg\thanks{
Department of Computer Science,
University of Texas at Dallas,
Box 830688,
Richardson, TX 75083
USA. E-mail: {\tt besp@utdallas.edu}
}
\and
Timothy Rozario\thanks{
Department of Computer Science,
University of Texas at Dallas,
Box 830688,
Richardson, TX 75083
USA. E-mail: {\tt tmr100020@utdallas.edu}
}
}

\date{}

\begin{document}
\maketitle

{\bf Abstract:}
In this paper we study problems of drawing graphs in the plane using 
edge length constraints and angle optimization. Specifically we consider the problem of  
maximizing the minimum angle, the {\em MMA problem}. 
We solve the {\em MMA} problem using a spring-embedding approach where two forces are applied to the vertices of the graph: a force optimizing edge lengths and a force optimizing angles. We solve analytically the problem of computing an optimal displacement of a graph vertex optimizing the angles between edges incident to it if the degree of the vertex is at most three. We also apply a numerical approach for  computing the forces applied to vertices of higher degree. We implemented our algorithm in Java and present drawings of some graphs.

\section{Introduction}

Angular resolution is one of the aesthetic criteria measuring the quality of graph drawings in terms of human comprehension. 
The {\em angular resolution} of a straight-line drawing in the plane is the minimum angle between any two incident edges.    
The study of graph drawing with angular resolution started by Formann \etal \cite{formann93} 
in 1990. They introduced the {\em angular resolution} of a graph as the supremum angular resolution over all straight-line drawings of the graph. The problem of computing the angular resolution of a graph is NP-hard (even the problem of deciding if $R=\pi/2$ for graphs with vertex degrees at most four is NP-hard).   

The main focus in the early investigation \cite{formann93,mp-arpg-94} was on bounding the angular resolution of a graph in terms of the maximum vertex degree $d$. The obvious upper bound is $R(d)\le 2\pi/d$. A lower bound $R(d)\ge \Omega(1/d)$ has been proved \cite{formann93}  for many graphs including planar graphs and complete graphs. 
A lower bound $R(d)\ge \Omega(1/d^2)$ holds for all graphs \cite{formann93}.
If we insist on planar straight-line drawings, then  $R(d)\ge \alpha^d$ for some constant $0<\alpha<1$ \cite{mp-arpg-94}. 

There was also a study of the optimization problem where the angular resolution of a given graph is maximized.
Matousek \etal \cite{msw-96} considered an {\em angle-optimal placement of point in polygon}. In this problem,  the task is to find a point $p$ in the kernel of a star-shaped polygon $P$ such that after connecting $p$ to all the vertices of $P$ by straight edges, the minimum angle 
between two adjacent edges is maximized. The {\em kernel} of $P$  is defined as the locus of 
all points inside $P$ that see all the edges and vertices of $P$ (it is not empty since $P$ is star-shaped). They showed that it is an LP-type problem of combinatorial dimension 3.
Amenta \etal \cite{abe-99} studied various problems of optimal point placement for mesh smoothing (using different mesh quality measures). A related result is the polynomial-time 
algorithm for computing a Steiner point in a star-shaped polygon, minimizing 
the maximum angle. A parallel algorithm for mesh smoothing is presented in 
\cite{mesh1}.

Carlson and Eppstein \cite{ce-tcf-06} considered tree drawings such that  all faces form  convex polygons (the infinite faces are created by extending the edges incident on leaves to the infinity). They showed that the optimal angular resolution can be computed in linear time and even the lengths of the tree edges may be chosen arbitrarily.
In the recent paper, Eppstein and Wortman \cite{ew-oar-11} considered 
graph drawing in the plane with faces drawn as centrally symmetric convex polygons. 
They found a polynomial time algorithm for computing a drawing maximizing the angular resolution. 

Another direction in graph drawing with the angle resolution is {\em  Lombardi drawing} where 
the graph edges are represented as circular arcs instead of straight-line segments 
\cite{ccgkt-11,degkn-10,degkn-12}.
Relaxing  the condition of straight-line segments allows to achieve {\em perfect angular resolution} where the edges are equiangularly spaced around each vertex. 
The classes of graphs admitting Lombardi drawings are presented in \cite{degkn-12} (and the algorithms for finding these drawings). 
Duncan \etal~\cite{degkn-10} found  that unrooted trees can drawn with perfect angular resolution and polynomial area.

In this paper we study the problem of maximizing the angular resolution using the force-directed approach.  This idea is not new and some algorithms for optimizing angular resolution using the force-directed or spring-embedding approach \cite{abs-mtrg-10,ccgkt-11}. The common feature of these approaches is that the forces  are directed from one vertex to another vertex. 
Our approach is different in the sense that we want to optimize the dislocation of a vertex. 
It can be viewed as a restricted version of the angle-optimal placement of point in polygon
where a graph is embedded in the plane with straight-line edges and we want to find a new position of a given vertex by moving it at distance at most $r>0$ (for a given $r$) and optimizing the incident angles. We call it {\em Max-Min Angle Problem} (MMA). 
This perturbation  problem is interesting in its own right. 
From combinatorial point of view, it is an LP-type problem with the same combinatorial 
dimension as the angle-optimal placement of point in polygon \cite{msw-96} since only one 
new constraint is added. However, the situation is quite different from the 
algebraic point of view.

The angle-optimal placement of point in polygon is related to the famous Fermat problem (which appears as a special case when the vertex degree is three). In general it is known as 
the Fermat-Weber problem: given $n$ points in a Euclidean space $\R^d$, find a point minimizing the sum of the distances to $n$ given points. This point is called the {\em geometric median}. We assume $d=2$. 
If $n\le 4$ then the geometric median can be computed exactly. However, it cannot be computed {\em exactly} if $n\ge 5$, in general (i.e. for some instances) \cite{b-ad-88}.  

The MMA problem is harder algebraically than the angle-optimal placement of point in polygon since the optimal point can be at distance $r$ from the initial point. 
In fact, we show that even for a vertex of degree three,  the solution involves polynomials of degree 6 in some (difficult) cases. The polynomials of degree at least five  cannot be solved exactly in general. The main result of this paper is that  the solution of the MMA problem for a vertex of degree three can be computed {\em exactly} (we show that it can be expressed using a polynomial of degree four only).

Another reason why we introduce and study the MMA problem is that the parameter $r$ allows to control the strength of the angular resolution force applied to vertices. This works well if we use more than one force. For example, we applied it to the spring embedding with length constraints. 

The paper is organized as follows.
In Section \ref{s2} we briefly describe force-directed graph drawing and introduce Max-Min-Angle problem. In Section \ref{s3} we recall the classical Fermat problem that appears as the special case of the MAA problem. In Section \ref{s:MAA} we provide a solution of the MMA problem for vertex degree two. 
Optimal solution for degree three vertices is provided in Section \ref{sectdeg3}. 
The algorithm and its performance is discussed in Section \ref{algorithm}.
Finally we conclude in Section \ref{concl}.

\section{Spring Embedder and Problem Statement} \label{s2}

Force-directed graph drawing is a popular technique and there is a growing literature on 
force-directed drawing algorithms, see the recent survey by Kobourov \cite{k-fdda-12}.
Eades \cite{EA84} introduced a mechanical model for graph drawing. 
To achieve aesthetically pleasing layouts and capture the edge length constraints he applied 
attracting/repelling force between two vertices if the distance between them is less/greater than the desired length. He found that the Hookes Law (linear) springs 
are too strong when the vertices are far apart but the logarithmic force solves this problem.
As initial embedding of the graph the algorithm places its vertices at random locations. 
The algorithm stops after a sufficient number of iterations. 
If a state of equilibrium is reached i.e. all forces are zero, then the graph embedding reaches 
the desired positioning in the plane and remains static. An example of such drawing is shown in Figure \ref{fig:algae}.

\begin{figure}[htp]
\centering
\includegraphics[scale=0.4]{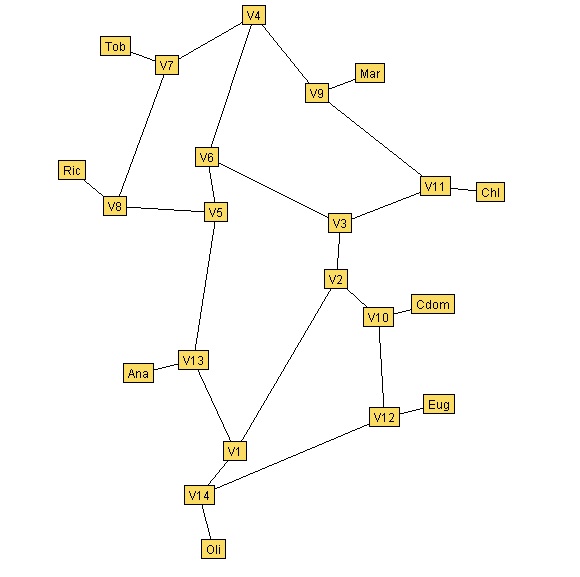}
\caption{A drawing of Algae graph \cite{pb-opn-12} produced by the spring embedder \cite{embedder}.}
\label{fig:algae}
\end{figure}

Fruchterman and Reingold \cite{ft91} added a condition of "even vertex distribution" which is modeled  by attractive forces between adjacent vertices and repulsive forces between all pairs of vertices. This increases the number of forces ($|V|^2$ repulsive forces for a graph $G=(V,E)$) and slows down the algorithm.

The algorithms by Eades  \cite{EA84} and Fruchterman and Reingold \cite{ft91} are just two examples of force-directed graph drawing. There are many spring embedders nowadays \cite{k-fdda-12}. 
We consider a new force-directed approach for angular resolution where the vertex displacement is optimized locally. 

\subsection{Problem Statement}

We consider one step of the spring embedder where the vertices of a
graph $G$ are embedded in the plane and they are allowed to move
slightly. The spring embedder may take into account several forces,
for example the forces aiming to achieve the desired lengths of the edges
of $G$. We introduce a force that aims to optimize, for all vertices
$v$, the angles between edges incident to $v$ and embedded in the
plane. Ideally, we want the edges to spread evenly around $v$. This
may not be possible due to other constraints of the drawing (edge
lengths, for example). Our goal is to maximize, for every $v$,  the
{\em smallest angle} between two edges incident to $v$. 
For this, we formulate our problem.

\begin{quote}
{\bf Max-Min-Angle (MMA) Problem}.
Let $G$ be a graph embedded in the plane. Let $P$ be the position
of a vertex $v$ of $G$ and let $A_1,A_2,\dots,A_k$ be the positions of the adjacent
vertices of $P$ in $G$. Let $r>0$ be a radius. Find a point $P^*$ (the
best position to move $P$ within distance $r$) such that $|PP^*|\le r$
and the smallest angle $\angle A_iP^*A_j$ is maximized.
\end{quote}

Let $Z$ be the circle centered at $P$ and radius $r$. 
We will solve the MMA problem in Sections \ref{s:MAA} and \ref{sectdeg3}
by considering degree $k$ of vertex $v$.

\section{Fermat problem} \label{s3}

When the degree of $v$ is three, the MMA  problem is related to 
the classical {\em Fermat problem}: Given triangle $ABC$ in the plane, find a point $F$ such that the total distance from the three vertices of the triangle to $F$ is the minimum possible. 
The solution of the Fermat problem (called {\em Fermat point} or {\em Torricelli point}) depends on triangle $ABC$.

\begin{figure}[htp] 
\centering
\includegraphics{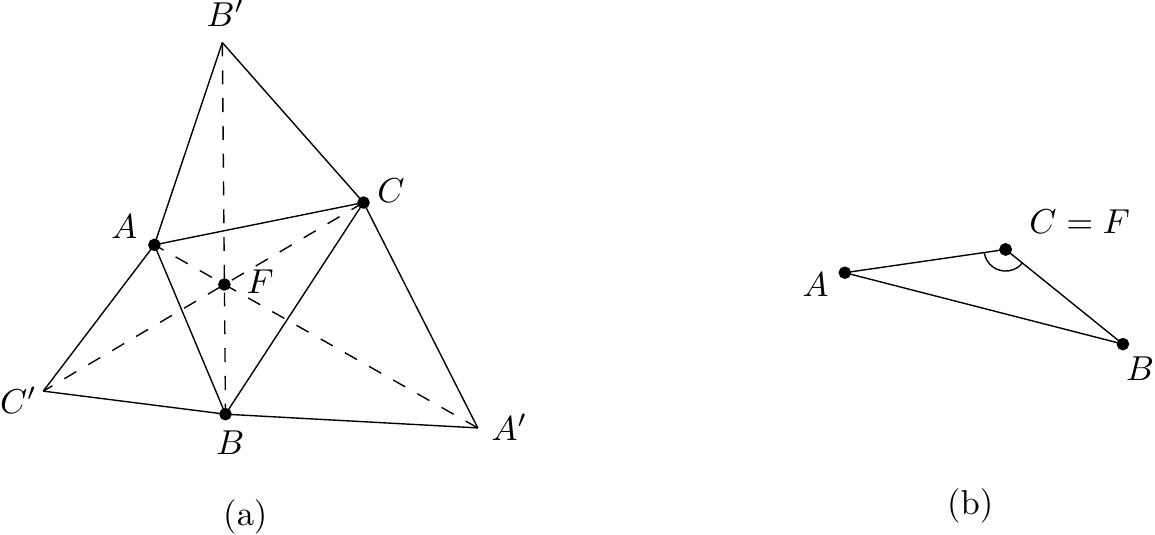}
\caption{The Fermat point $F$ in (a) Case 1 and (b) Case 2.
(a) Three angles $\angle AFB, \angle BFC$, and $\angle CFA$ are equal.
(b) The Fermat point is at $C$.}
\label{pf1}
\end{figure}

 {\bf Case 1:}  All angles of triangle $ABC$ are smaller than $120^\circ$.\\
Construct three new regular triangles $ABC', AB'C,$ and $A'BC$ out of the three sides of triangle $ABC$.
Then the point $F$ is the intersection of three lines $AA',BB'$, and $CC'$, see Figure \ref{pf1}(a). 
In this case the Fermat point $F$  is coincident with the {\em first isogonic center} 
of the triangle \cite{Eric}, where the angles subtended by the sides of the triangle are all equal, 
i.e.  $\angle AFB, \angle BFC, \angle CFA=120^\circ$.

{\bf Case 2:}  There exists an angle  in triangle $ABC$ greater than or equal to $120^\circ$.\\
Only one angle of the triangle is greater than or equal to $120^\circ$, say $\angle ACB\ge 120^\circ$, see Figure \ref{pf1}(b) for example.
Then the Fermat point is coincident with $C$.

\section{The MMA problem} \label{s:MAA}

We solve the MMA problem depending on the degree of vertex $v$. 

\subsection{Vertex of Degree 2} \label{sectdeg2}

Suppose that the degree of $v$ is two.  
Let $A$ and $B$ be the positions of its adjacent vertices. 
The task is to find a point $P^*$ with $|PP^*|\le r$ maximizing angle $\theta=\min(\angle AP'B,\angle BP'A)$ (the angles are in the counterclockwise order).
Suppose that segment $AB$ intersects circle  $Z$.
Obviously, $\theta=180^\circ$ is achieved by any point $P^*$ in the intersection of $AB$ and $Z$.
The solution is unique if $AB\cap Z$ is a single point (for example, $AB$ may be tangent to $Z$).

\begin{figure}[htp] 
\centering
\includegraphics{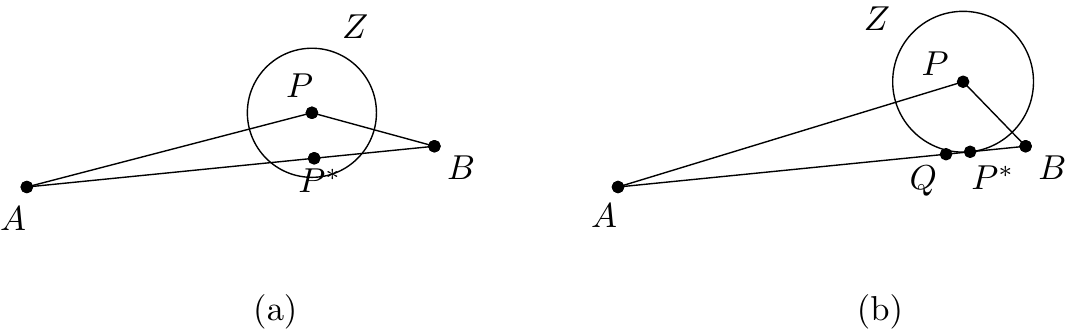}
\caption{ The case where $AB\cap Z\ne\emptyset$. (a) $P^*=Q$ where $Q$ is defined by Equation(\ref{deg2a}). (b) Point $Q$ lies not in $AB\cap Z$.}
\label{f2}
\end{figure}

Suppose that $AB\cap Z$ is a segment of positive length, see Figure \ref{f2}(a).
There are infinitely many positions of $P^*$ (all in $AB\cap Z$ achieving $\theta=\pi$).
If $P$ is not on segment $AB$ then any location of $P^*$ will change distances to $A$ and $B$.
We would like to preserve the ratio $|AP|/|BP|$ and compute point $Q$ in segment $AB$ 
such that $|AQ|/|BQ|=|AP|/|BP|$. This implies

\begin{equation} \label{deg2a}
\frac{|AQ|}{|AB|}=\frac{|AP|}{|AP|+|BP|}.
\end{equation}

Then the coordinates of point $Q$ can be computed using

\begin{align}
Q_x&=A_x+ \frac{|AP|}{|AP|+|BP|} (B_x-A_x), \label{qx}\\
Q_y&=A_y+ \frac{|AP|}{|AP|+|BP|} (B_y-A_y). \label{qy}
\end{align}

If point $Q$ lies in circle $Z$ then $P^*=Q$ as shown in Figure \ref{f2}(a).
Otherwise we select the endpoint of segment $AB\cap Z$ that is closer to $Q$, see Figure \ref{f2}(b) for an example.

We assume now that $AB$ and $Z$ do not intersect.

\begin{prop} \label{propdeg2}
If segment $AB$ and circle $Z$ do not intersect then $P^*$ is on the boundary of $Z$.
Furthermore the circle passing through $A,B$, and $P^*$ is tangent to $Z$.
\end{prop}

\begin{proof}
Suppose $P^*=(x,y)$ where $x$ and $y$ are unknown.
Consider circular arc $AP^*B$ as shown in Figure \ref{f1}(a).
The angles $\angle AP'B$ are equal for all points $P'$ from arc $AP^*B$ (since the inscribed angle of a chord equals half of the central angle ,see Figure \ref{pointInCirc}(b)). Therefore we can assume that $P^*$ is on the boundary of $Z$. Furthermore arc  $AP^*B$ cannot intersect the boundary of $Z$ twice since any point $P'$ from the arc of $Z$ cut off by arc $AP^*B$ makes an angle larger than $P^*$ makes, i.e.
$\angle AP'B>\angle AP^*B$ as shown in Figure \ref{f1}. Thus, arc $AP^*B$ must be tangent to $Z$ as shown in Figure \ref{f3}(a).
\end{proof}

\begin{figure}[htp] 
\centering
\includegraphics{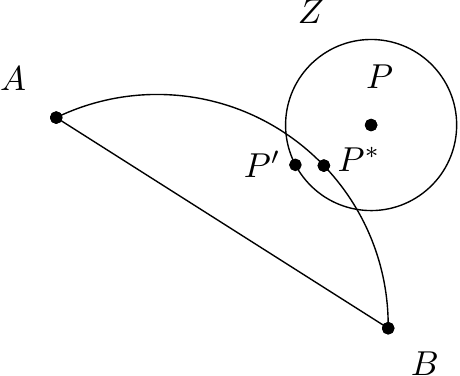}
\caption{Arc $AP^*B$. Wrong location of $P^*$ since $\angle AP'B>\angle AP^*B$.}
\label{f1}
\end{figure}

The MMA problem can be formulated now as the following problem dealing with only one angle.

\begin{quote}
{\bf Maximum Angle Problem}.
Let $A,B$, and $P$ be points in the plane and a number $r>0$ such that $|AP|,|BP|>r$.
Compute a point $P^*$ with $|PP^*|=r$ maximizing angle
$\angle AP^*B$.
\end{quote}

\begin{figure}[htp]
\centering
\includegraphics{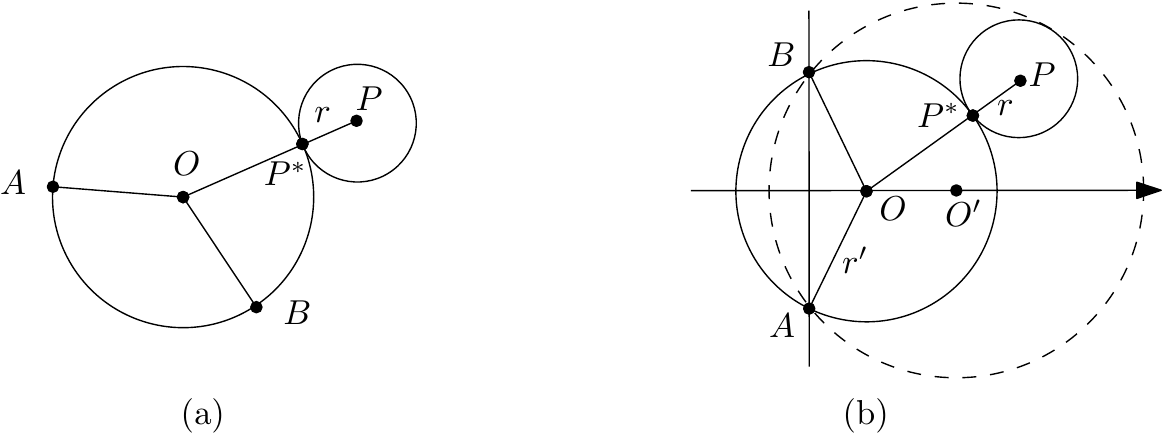}
\caption{(a) The maximum angle problem. (b) The circles centered at $O$ and $O'$ are the two solutions of the system of equations (\ref{map1}) and (\ref{map2}). }
\label{f3}
\end{figure}

Let $O$ be the center of the circle passing through $A,B$, and $P^*$.
It suffices to find point $O$ since point $P^*$ is the intersection point of the circle $Z$ and the line passing through points $O$ and $P$. The first equation is $(P^*_x-P_x)^2+(P^*_y-P_y)^2=r^2$ and the second equation is linear
\[ \left|
\begin{array}{ccc}
    P^*_x & P^*_y          & 1 \\
    P_x & P_y          & 1 \\
    O_x & O_y          & 1 
\end{array} \right|=0, \]
where the subscripts denote the coordinates.

Using coordinate transformations we can assume without loss of generality that $A(0,-1)$ and $B(0,1)$, see Figure \ref{f3}(b). Then $|OA|=|OB|$ implies that $O(x,0)$ where $x$ is unknown. Let $r'=|OA|$ be the radius of the circle centered at $O$.  
Then $|OP|=r+r'$ and $|OA|=r'$. They can be written as

\begin{empheq}[left=\empheqlbrace]{align}
(x-P_x)^2 + P_y^2 &= (r+r')^2 \label{map1} \\
x^2 + 1 &= r'^2 \label{map2} 
\end{empheq}

Subtracting Equation (\ref{map1}) from Equation (\ref{map2}) , we obtain
\begin{equation} \label{eqx}
2P_x x = P_x^2+P_y^2-2r'r-r^2-1.
\end{equation}

By plugging $x$ from this equation into Equation of (\ref{map2})  we obtain a quadratic equation in the variable $r'$. There are two roots of the quadratic equation and they correspond to two circles shown in Figure \ref{f4}. Therefore we take the smallest root of the quadratic equation.

\begin{figure}[htp]
\centering
\includegraphics{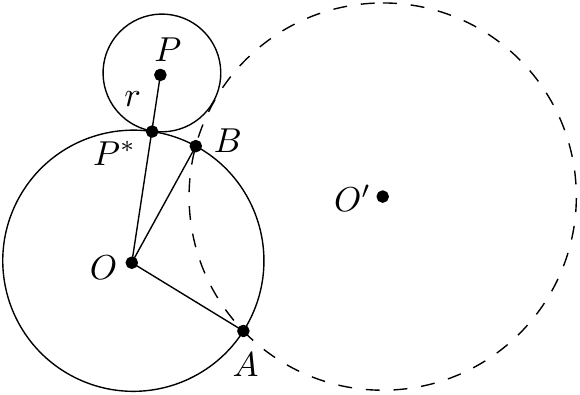}
\caption{An example of the maximum angle problem. The circles centered at $O$ and $O'$ are the two solutions of the system of equations (\ref{map1}) and (\ref{map2}). }
\label{f4}
\end{figure}

Therefore we proved the folowing claim.
\begin{prop} \label{propdeg2a}
The maximum angle problem can be solved using a quadratic equation.
\end{prop}

\section{Vertex of degree 3}  \label{sectdeg3}

In this Section we consider the case where the degree of $v$ is three. 
Let $A,B$ and $C$ be the positions of its adjacent vertices in $G$.
We consider the first case of the Fermat point of triangle $ABC$ where all angles of $ABC$ are less than $120^{\circ}$.

\begin{prop} \label{prop2}
If all angles of $ABC$ are less than $120^{\circ}$ and 
the Fermat point $F$ of $ABC$ lies within circle $Z$, then $P^*=F$.
\end{prop}

\begin{proof}
Without loss of generality we assume that 
$A,B,C$ are in counterclockwise order around $P^*$, see Figure \ref{Fpoint} (a).
Suppose to the contrary that $P^*$ is a point inside $Z$ different from $F$ as shown in Figure \ref{Fpoint} (b).
The Fermat point lies in one of triangles $AP^*B,  BP^*C$,  or $AP^*C$.
Suppose that it belongs to triangle $AP^*C$. 
Then $\angle AP^*C < \angle AFC =120^{\circ}$ and $P^*$ is not the optimal point. 
Therefore $P^*=F$, see Figure \ref{Fpoint} (c).
\end{proof}

\begin{figure}[htp] 
\centering
\includegraphics{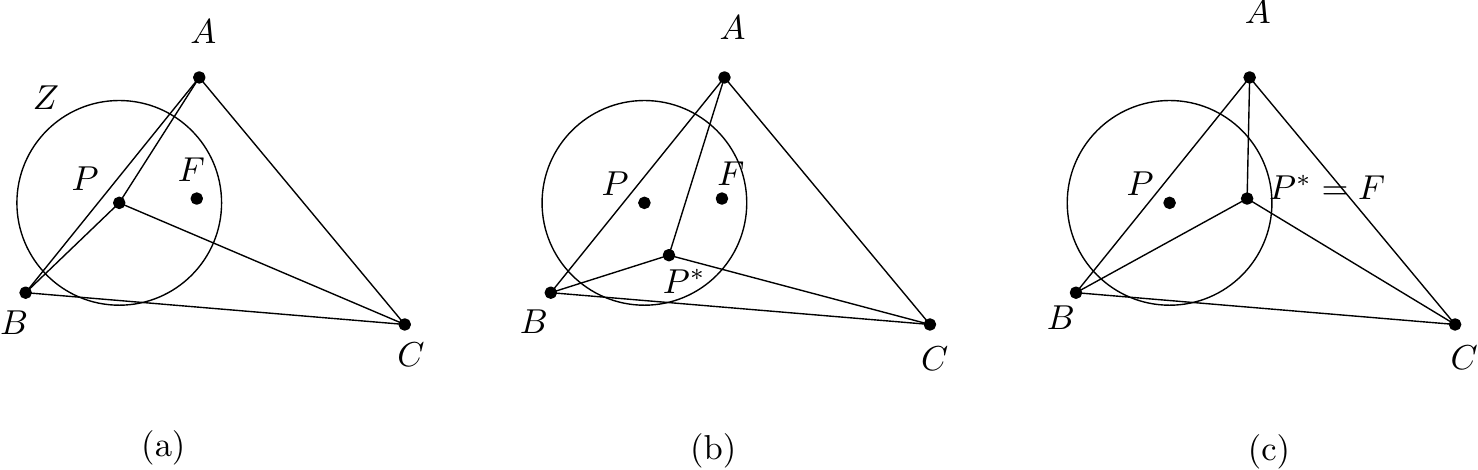}
\caption{(a) Setting of Proposition \ref{prop2}. (b) Contradiction when $P^*\neq F$.  (c)  Points  $P^*$ and $F$ coincide.}
\label{Fpoint}
\end{figure}

Note that in the case of Proposition \ref{prop2} all angles around $P^*$ are equal 
$120^{\circ}$.  
In the rest of this section we consider cases where not all angles 
$\angle AP^*B,  \angle BP^*C$,  and $\angle AP^*C$ are equal. 
If the smallest angle of $\angle AP^*B,  \angle BP^*C$,  and $\angle AP^*C$
is unique, say it is $\angle AP^*B$, then it can be computed by solving the maximum angle problem for $A,B$ using the method from the previous section.
It remains to consider the case of two smallest angles.

\subsection{Two smallest angles} \label{twomin}

In this section we assume that there are exactly two smallest angles from $\{\angle AP^*B, \angle  AP^*C, \angle BP^*C\}$. First, we show that $P^*$ must be on the boundary of circle $Z$. 

\begin{lemma} \label{lemma2angles}
If the degree of $v$ is 3 and there are two smallest angles around $P^*$ then $P^*$ lies on the boundary of $Z$. 
\end{lemma}

\begin{proof}
We assume that \\
(i) $A,B,C$ are in counterclockwise order around $P^*$, and \\
(ii) $\angle AP^*B=\angle BP^*C<\angle CP^*A$.\\
Suppose to the contrary that $P^*$ lies inside $Z$ as shown in Figure \ref{pointInCirc}.
Draw two circles, one passing through points $A,B,P^*$ and the second passing through points $B,C,P^*$. Since these circles intersect by two points $B$ and $P^*$, then their interiors intersect by a lune. 
Since $P^*$ is inside circle $Z$, then the intersection of the lune and the interior of $Z$ is  a non-empty set $I$, see the shaded area in Figure \ref{pointInCirc}(a).
For any point $P'$ in $I$, $\angle AP'B>\angle AP^*B$ and $\angle BP'C>\angle BP^*C$ 
(this can be seen, for example, by the fact that the inscribed angle of a chord equals half of the central angle, see Figure \ref{pointInCirc}(b)).
We choose $P'$ from $I$ close enough to $P^*$ such that $\angle AP'B,\angle BP'C<\angle CP'A$. Then the smallest angle around $P'$ is greater than $\angle AP^*B$. 
Therefore $P^*$ is not the solution of the MMA problem. Contradiction.
\end{proof}
 
\begin{figure}[htp]
\centering
\includegraphics{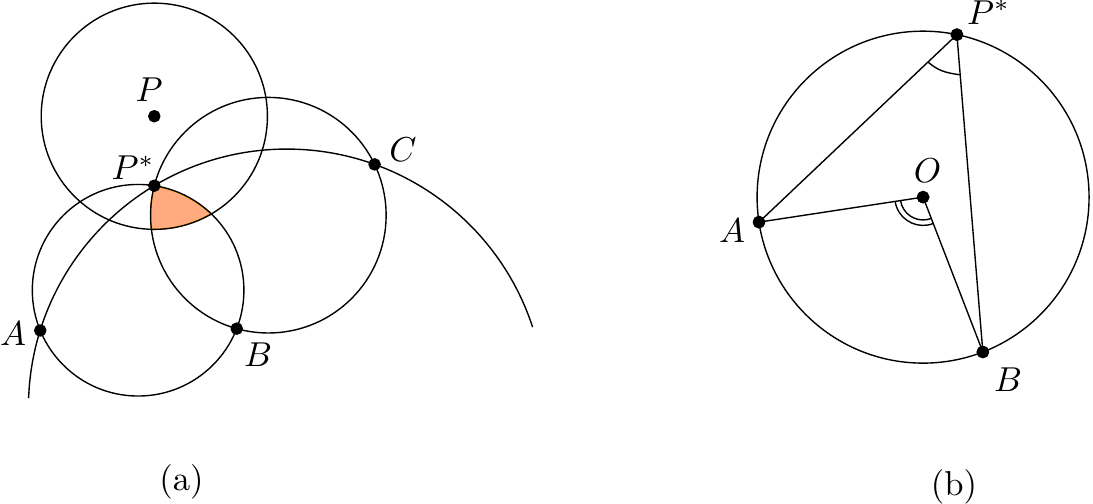}
\caption{(a) For all points $P'$ in the shaded area $\angle AP'B>\angle AP^*B$ and $\angle BP'C>\angle BP^*C$.  (b)  The inscribed angle $\angle AP^*B$ of chord $AB$ equals half of the central angle $\angle AOB$.}
\label{pointInCirc}
\end{figure}

The main result in this section is the following

\begin{thm} \label{thm2angles}
If the degree of vertex $v$ is three and there are two smallest angles around $P^*$ then $P^*$ can be computed by solving a polynomial equations of degree at most four.
\end{thm}

\begin{figure}[htp]
\centering
\includegraphics[scale=1]{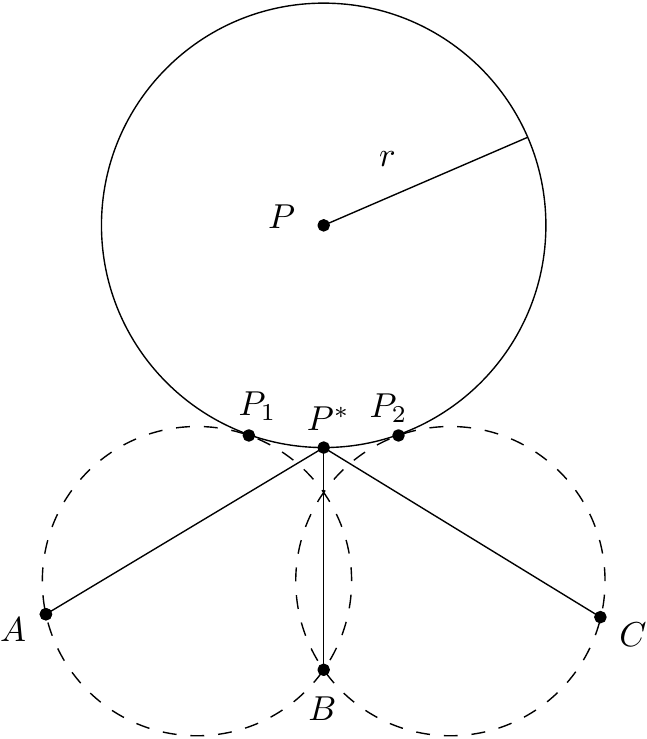}
\caption{
Points $P_1$ and $P_2$ are the points on circle $Z$ maximizing angles $\angle AP_1B$ and $\angle BP_2C$, respectively. Point $P^*$ is different from $P_1$ and $P_2$ and angles $\angle AP^*B=\angle BP^*C$.}
\label{Ipe2Angle}
\end{figure}

\begin{proof}
First we show that it is possible that there are two smallest angles. 
An example where two smallest angles around $P^*$ are equal is shown in Figure  \ref{Ipe2Angle} where points $P,P^*$ and $B$ are collinear and 
points $A$ and $C$ are symmetric about the line $PB$.
Therefore angles $\angle AP_2B<\angle AP^*B<\angle AP_1B$ and $P^*$ is the solution of the MMA problem.

We now prove the theorem. Suppose that $\angle AP^*B=\angle BP^*C<\angle CP^*A$. 
We need to find coordinates of $P^*(x,y)$.
By the law of cosines we write the equation $\cos(\angle AP^*B)=\cos(\angle BP^*C)$ as

\begin{gather}
\frac {|BP^*|^2+|AP^*|^2-|AB|^2}{2\cdot|BP^*|\cdot|AP^*|}=\frac{|BP^*|^2+|CP^*|^2-|BC|^2}{2\cdot|BP^*|\cdot |CP^*|}\\
|CP^*|\cdot (|BP^*|^2+|AP^*|^2-|AB|^2)=|AP^*|\cdot (|BP^*|^2+|CP^*|^2-|BC|^2) \label{eq43} 
\end{gather}

Without loss of generality we can assume that $P=(0,0)$ and $A_y=C_y$. 
By Lemma \ref{lemma2angles} we assume $|PP^*|=r$. 
Then $x^2+y^2=r^2$ and
\begin{gather}
|CP^*|^2=r^2+C_x^2+C_y^2-2C_xx-2C_yy=a_1x+a_2y+a_3\\
\frac 12 (|BP^*|^2+|AP^*|^2-|AB|^2) =
r^2+A_xB_x+A_yB_y-(A_x+B_x)x-(A_y+B_y)y
=a_4x+a_5y+a_6\\
|AP^*|^2=r^2+A_x^2+A_y^2-2A_xx-2A_yy=b_1x+b_2y+b_3\\
\frac 12 (|BP^*|^2+|P^*C|^2-|BC|^2)=
r^2+B_xC_x+B_yC_y-(B_x+C_x)x-(B_y+C_y)y
=b_4x+b_5y+b_6,
\end{gather}
where $a_1=-2C_x, a_2=-2C_y, a_3=r^2+{C_x}^2+{C_y}^2,a_4=-(A_x+B_x), a_5=-(A_y+B_y)$,\\
$a_6=B_xA_x+B_yA_y+r^2$, 
$b_1=-2Ax, b_2=-2A_y, b_3=r^2+{A_x}^2+{Ay}^2, b_4=-(B_x+C_x), b_5=-(By+Cy)$,\\ $b_6=B_xC_x+B_yC_y+r^2$.

Take the square of both sides of Equation (\ref{eq43})
\begin{gather}
(a_1x+a_2y+a_3)(a_4x+a_5y+a_6)^2=
(b_1x+b_2y+b_3)(b_4x+b_5y+b_6)^2 \label{eq57}
\end{gather}
It can be written as
\begin{gather}
c_{1}x^3+c_{2}y^3+ c_{3}x^2+c_{4}y^2+c_{5}x^2y+c_{6}xy^2+c_{7}x+c_{8}y+c_{9}xy+c_{10}=0 
\label{eq-c}
\end{gather}
where $c_{1}=a_1a^2_4-b_1b^2_4, c_{2}=a_2a^2_5-b_2b^2_5, c_{3}=(2a_1a_4a_6+a_3a^2_4-(2b_1b_4b_6+b_3b^2_4))$, \\
$c_{4}=(2a_2a_5a_6+a_3a^2_5-(2b_2b_5b_6+b_3b^2_5)), c_{5}=(2a_1a_4a_5+a_2a^2_4-(2b_1b_4b_5+b_2b^2_4))$,\\
$c_{6}=(2a_2a_4a_5+a_1a^2_5-(2b_2b_4b_5+b_1b^2_5)), c_{7}=(2a_3a_4a_6+a_1a^2_6-(2b_3b_4b_6+b_1b^2_6))$, \\
$c_{8}=(2a_3a_5a_6+a_2a^2_6-(2b_3b_5b_6+b_2b^2_6)), c_{9}=2(a_1a_5a_6+a_2a_4a_6+a_3a_4a_5-(b_1b_5b_6+b_2b_4b_6+b_3b_4b_5))$,\\
$c_{10}=a_3a^2_6$.

Using $y^2=r^2-x^2$ we reduce powers of $y$ in Equation (\ref{eq-c})
\begin{gather}
c_{1}x^3+y(c_2r^2+c_8)+c_3x^2+c_4-c_4x^2+x^2y(c_5-c_2)+
c_6x-c_6x^3+ c_7x+c_8y+c_9xy+c_{10}=0\\
d_1x^3+d_2x^2+d_3x+d_4
+(d_5x^2+d_6x+d_7)y=0 \label{eq49}
\end{gather}
where
$d_1=c_1-c_6, d_2=c_3-c_4,d_3=c_6r^2+c_7,d_4=c_4r^2+c_{10},
d_5=c_5-c_2,d_6=c_9,d_7=c_2r^2+c_8$.

\begin{gather}
\left( d_1x^3+d_2x^2+d_3x+d_4 \right)^2 - ((d_5x^2+d_6x+d_7)^2\cdot(r^2-x^2))=0 \label{eq11}
\end{gather}

Simplifying Equation (\ref{eq11}) we get a sextic equation

\begin{equation}
e_1x^6+e_2x^5+e_3x^4+e_4x^3+e_5x^2+e_6x+e_7=0 \label{eq12}
\end{equation}
where
$e_1=d_1^2+d_5^2, 
e_2=2(d_1d_2+d_5d_6),
e_3=(d_2^2+2d_1d_3-(d_5^2r^2-(d_6^2+2d_5d_7)))$,\\
$e_4=(2(d_1d_4+d_2d_3)-(2(d_5d_6r^2-d_6d_7))),e_5=((d_3^2+2d_2d_4)-((d_6^2+2d_5d_7)r^2-d_7^2))$,\\ 
$e_6=2d_3d_4-(2d_6d_7r^2),e_7=d_4^2-d_7^2r^2$.

By Lemma \ref{factor} the sextic equation has a quadratic factor. 
Dividing by it, we obtain a polynomial of degree four (Equation (\ref{eqn1}),(\ref{eqn2}), or
(\ref{eqn3})) and the theorem follows.
\end{proof}

\begin{lemma} \label{factor}
The polynomial equation (\ref{eq12}) has factor $x^2+A_y^2-r^2$.
\end{lemma}

\begin{proof}
We consider 3 cases.

Case I: $A_y=r$.
Then $e_5=e_6=e_7=0$ and the polynomial (\ref{eq12}) 
has factor $x^2$. The polynomial is reduced to the polynomial of degree 4
\begin{equation} \label{eqn1}
e_1x^4+e_2x^3+e_3x^2+e_4x=0.
\end{equation}

Case II: $A_y^2-r^2>0$.
We scale the coordinates such that $A_y^2-r^2=1$.
Then $e_6=e_4-e_2$ , $e_7=e_3-e_1$ and the polynomial (\ref{eq12}) 
has factor $x^2 + 1$. The polynomial is reduced to the polynomial of degree 4

\begin{equation} \label{eqn2}
e_1x^4+e_2x^3+(e_3-e_1)x^2+(e_4-e_2)x+e_5-(e_3-e_1)=0.
\end{equation}

Case III: $A_y^2-r^2<0$.
We scale the coordinates such that $A_y^2-r^2=-1$.
Then $e_6=e_4+e_2$ , $e_7=e_3+e_1$ and the polynomial (\ref{eq12}) 
has factor $x^2 - 1$. The polynomial is reduced to the polynomial of degree 4
\begin{equation} \label{eqn3}
e_1x^4+e_2x^3+(e_3+e_1)x^2+(e_4+e_2)x+e_5+(e_3+e_1)=0.
\end{equation}
The lemma follows.
\end{proof}

\section{Algorithm} 
\label{algorithm}

In this Section we discuss how to modify the spring embedder in order to take into account 
the angles between embedded edges. 
Let $G=(V,E)$ be an input graph. 
At the initial step the spring embedder randomly places the vertices of $G$ in the plane.
Then, it iterates a simultaneous motion of the vertex positions based on
one or several forces (springs). 
We describe a new force using angle optimization.
For every vertex $v$ with a position $P$ in the plane, the algorithm computes 
a new position $P^*$ and uses vector $PP^*$ as a force applied to vertex $v$.

{\bf Angle Optimization Algorithm}\\
Input: Graph $G=(V,E)$ embedded in the plane and radius $r$.\\
Output: New embedding of graph $G$ where each vertex is translated within distance $r$.

For each vertex $v\in V$ do the following steps.
\begin{enumerate}
\item Let $P$ be the position of $v$ in the plane. We compute $P^*$ as follows. Every time  $P^*$ is assigned, we proceed to the next vertex $v$.
\item If the degree of $v$ is at most one then set $P^*=P$.
\item Suppose that the degree of $v$ is equal to two. Let $A$ and $B$ be the positions of vertices adjacent to $v$. 
\begin{enumerate}
\item
Suppose $AB\cap Z\ne\emptyset$. Compute point $Q$ using Equations (\ref{qx}) and (\ref{qy}).  
If $|PQ|\le r$ then $P^*=Z$; otherwise assign $P^*$ to the endpoint of segment $AB\cap Q$ that is closer to $Q$.
\item If $AB\cap Z=\emptyset$ then set compute $P^*$ 
using Proposition \ref{propdeg2a}. 
\end{enumerate}
\item Suppose that the degree of $v$ is equal to three. 
Let $A,B$, and $C$ be the positions of vertices adjacent to $v$.
Compute Fermat point $F$ of triangle $ABC$.
If $|PF|\le r$ then set $P^*=F$. Otherwise, 
for each segment $ab\in \{AB,AC,BC\}$, compute point $P_{ab}$ maximizing angle 
$\angle aP_{ab}b$ (Section \ref{sectdeg2}). If angle $\angle aP_{ab}b$ is the smallest angle from
$\{ \angle AP_{AB}B, \angle AP_{AC}C, \angle BP_{BC}C \}$ then set $P^*=P_{ab}$. 
Otherwise compute $P^*$ as the best solution using two smallest angles for every two segments from $\{ AB,AC,BC\}$ (Section \ref{twomin}). 
\item 
For the remaining vertices of degree at least four, apply the following grid approach.
Pick a grid stepsize $\delta$, for example $\delta=r/3$. 
Consider a grid with the origin at $P$. For every grid point $Q$ with $|PQ|\le r$,
compute the smallest angle $\alpha_Q$ if $P$ is moved to $Q$. Find $Q$ maximizing 
$\alpha_Q$. Assign $P^*=Q$.
\end{enumerate}

We implemented this algorithm\footnote{Demo is available at 
\url{http://www.utdallas.edu/~sxb027100/soft/AngleOpt/}.} and run it on several graphs.
First, we tested the algorithm on graph $T_{10}$ from \cite{rv-06} since it contains vertices of degree two.  The program draws $T_{10}$ with angles optimized, see Figure \ref{t10} (b). 
It can be compared with the drawing of the original embedder \cite{embedder}, see
Figure \ref{t10} (a). 

\begin{figure}[htp]
\centering
\begin{minipage}[b]{0.45\linewidth}
\centering \includegraphics[scale=0.4]{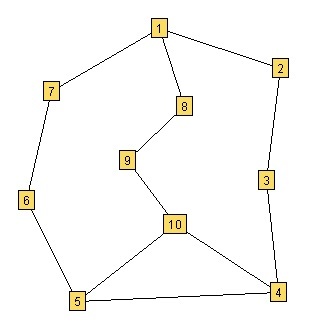}\\ (a)
\end{minipage}
\begin{minipage}[b]{0.45\linewidth}
\centering \includegraphics[scale=0.4]{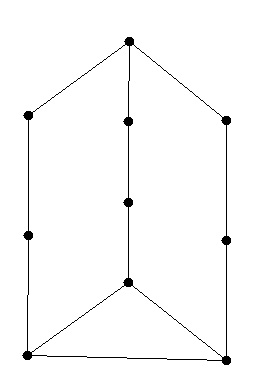}\\ (b)
\end{minipage}
\caption{Graph $T_{10}$ from \cite{rv-06}. (a) Drawing by Spring Embedder \cite{embedder}.
(b) Drawing by our program with angle optimization.}
\label{t10}
\end{figure}

We also tested the algorithm on the phylogenetic networks for Algae example \cite{pb-opn-12}.
The drawing of the Algae network by  the spring embedder \cite{embedder} is shown in Figure 
\ref{fig:algae}. It can be compared with our drawings in Figure \ref{algae}. 
In all drawings (in Figures \ref{fig:algae} and Figure \ref{algae}) the edge length constraints are
satisfied but the angle resolution in drawings in Figure \ref{algae} is significantly larger. 
The drawing shown in Figure \ref{algae} (a) uses the weighted version of the graph and the drawing in Figure \ref{algae} (b) uses the graph with  intermediate points on the edges. 

\begin{figure}[htp]
\centering
\begin{minipage}[b]{0.45\linewidth}
\centering \includegraphics[scale=0.4]{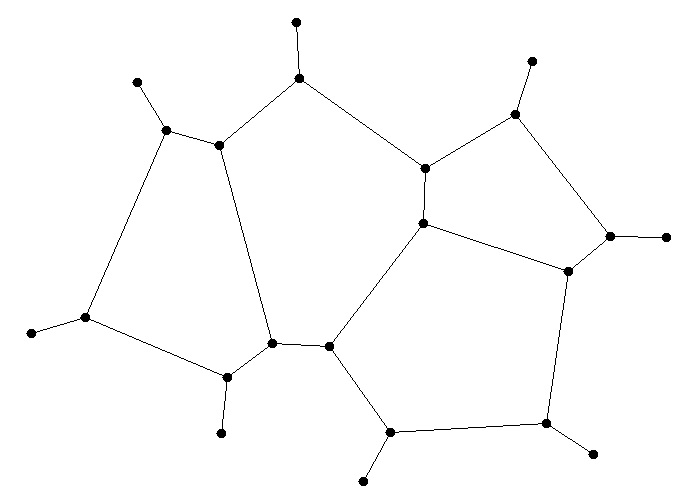}\\ (a)
\end{minipage}
\begin{minipage}[b]{0.45\linewidth}
\centering \includegraphics[scale=0.4]{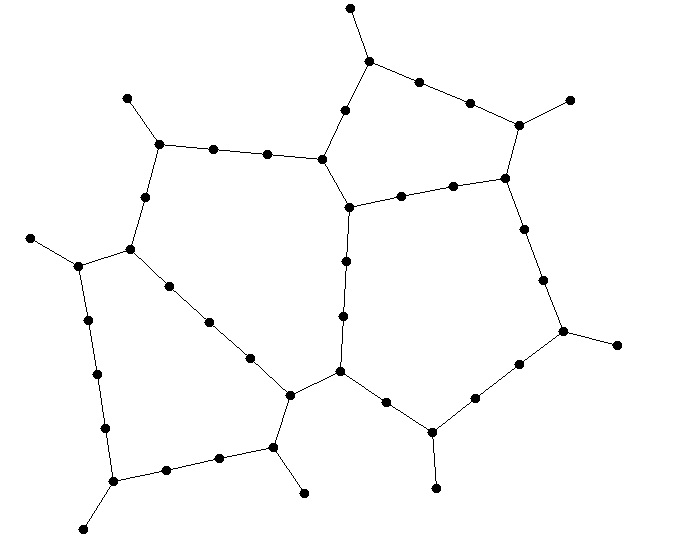}\\ (b)
\end{minipage}
\caption{Phylogenetic network for the Algae example \cite{pb-opn-12}.
(a) Weighted graph. (b) Graph with intermediate points on the edges.}
\label{algae}
\end{figure}

Finally, we run our program on the well-known graphs: 
the Petersen graph \cite{hs-93}, the Heawood graph \cite{har-91}, and 
the Herschel graph \cite{bh-81}, see Figure \ref{well}.
The Petersen graph is drawn with two crossings (the Petersen graph is in fact the 
smallest 2-crossing cubic graph), see Figure \ref{well}(a).
The Heawood graph has crossing number three (it is actually the smallest 3-crossing cubic graph) and our program found a drawing with exactly three crossings, see Figure \ref{well}(b).

\begin{figure}[htp]
\centering
\begin{minipage}[b]{0.31\linewidth}
\centering \includegraphics[scale=0.4]{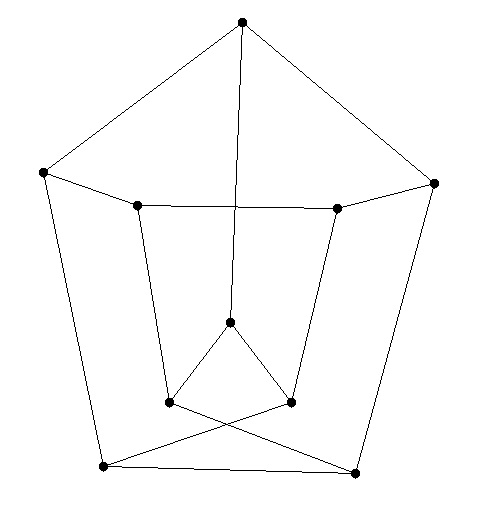}\\ (a)
\end{minipage}
\begin{minipage}[b]{0.31\linewidth}
\centering \includegraphics[scale=0.4]{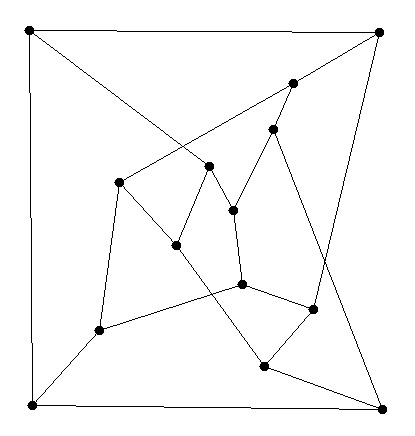}\\ (b)
\end{minipage}
\begin{minipage}[b]{0.31\linewidth}
\centering \includegraphics[scale=0.4]{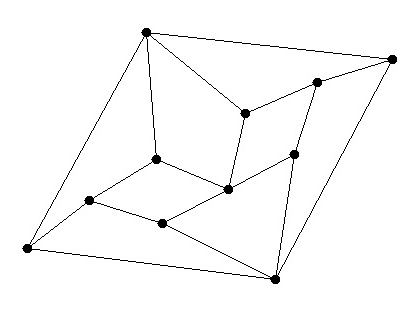}\\ (c)
\end{minipage}
\caption{(a) The Petersen graph. (b) The Heawood graph.
(c) The Herschel graph.}
\label{well}
\end{figure}

\section{Conclusion} \label{concl}

We proposed a novel approach to the problem of optimizing the angular resolution of a drawing. 
It has been applied to the spring embedder and the results with good angular resolution were obtained. It is known that the running time of the spring embedder can increase with the size of the graph. Therefore it is important to perform better the intermediate steps. The optimal solution of the MMA problem provides an opportunity to decrease the number of iterations of the spring embedder. 

The main result of this paper states that a vertex of degree at most three can be displaced optimally by solving a polynomial equation of degree at most four (which is interesting since the straightforward approach leads to a polynomial of degree six). A special case of MMA problem for degree three verices is the classical Fermat problem and the Fermat point is the solution for the special case. An interesting question for future research is to find the algebraic complexity of MMA problem for higher vertex degrees.

\bibliographystyle{plain}

\end{document}